\documentclass[onecolumn,11pt,accepted=2017-06-19]{quantumarticle}

\pdfoutput=1
\usepackage[latin9]{inputenc}
\usepackage[active]{srcltx}
\usepackage[english]{babel}
\usepackage{amsmath}
\usepackage{amsthm}
\usepackage{amssymb}
\usepackage{amsthm}

\theoremstyle{plain}

\makeatletter

\makeatletter
\newtheorem*{rep@theorem}{\rep@title}
\newcommand{\newreptheorem}[2]{%
\newenvironment{rep#1}[1]{%
 \def\rep@title{#2 \ref{##1}}%
 \begin{rep@theorem}}%
 {\end{rep@theorem}}}
\makeatother

\newtheorem{theorem}{Theorem}

\newtheorem*{theorem*}{Theorem}
\newreptheorem{theorem}{Theorem}

\theoremstyle{plain}

\theoremstyle{definition}
\newtheorem{defn}[theorem]{\protect\definitionname}
\theoremstyle{plain}
\newtheorem{lem}[theorem]{\protect\lemmaname}
\theoremstyle{remark}

\theoremstyle{remark}
\newtheorem*{rem*}{\protect\remarkname}
\theoremstyle{remark}

\theoremstyle{plain}

\usepackage{fullpage}

\usepackage{amsmath}

\usepackage{authblk}

\usepackage{xcolor}

\definecolor{dgreen}{rgb}{0.1,0.5,0.1}

\DeclareMathOperator{\poly}{poly}

\newcommand{\ket}[1]{|#1\rangle}
\newcommand{\bra}[1]{\langle#1|}
\newcommand{\eps}{\epsilon}

\usepackage{hyperref}

\makeatother

\providecommand{\claimname}{Claim}
\providecommand{\definitionname}{Definition}
\providecommand{\lemmaname}{Lemma}
\providecommand{\remarkname}{Remark}

\newcommand{\cgscon}{\ensuremath{\mathsf{CGSCON}}}
\begin{document}

\title{QCMA hardness of ground space connectivity for commuting Hamiltonians}

\author[1,2]{David Gosset}
\author[1]{Jenish C. Mehta}
\author[1]{Thomas Vidick}
\affil[1]{California Institute of Technology}
\affil[2]{IBM T.J Watson Research Center}

\date{\today}

\maketitle

\begin{abstract}

In this work we consider the ground space connectivity problem for  commuting local Hamiltonians. The ground space connectivity problem
asks whether it is possible to go from one (efficiently preparable) state to
another by applying a polynomial length sequence of 2-qubit unitaries while remaining at all times in a state with low energy for a given Hamiltonian $H$. It was shown in \cite{gharibian2015ground}
that this problem is $\mathsf{QCMA}$-complete for general local Hamiltonians, where $\mathsf{QCMA}$ is defined as $\mathsf{QMA}$ with a classical witness and $\mathsf{BQP}$ verifier. Here we show that the commuting version of the problem is also $\mathsf{QCMA}$-complete. This provides one of the first examples where commuting local Hamiltonians exhibit complexity theoretic hardness equivalent to general local Hamiltonians. 
\end{abstract}

\section{Introduction}

Since the definition and development of $\mathsf{\mathsf{NP}}$ completeness
in the 1970s \cite{cook1971complexity}, constraint satisfaction problems and their hardness
have been intensively studied. The following problem, called the $k$-constraint satisfaction
problem ($k$-CSP), generalizes many well-known $\mathsf{NP}$-complete
problems: Given a set of $n$ variables each taking values in a finite
set $S$, and given a set of constraints on the values that any
$k$-tuple of variables can take, is there an assignment to the
variables such that all constraints are satisfied? A quantum variant of the $k$-CSP, called the $k$-local Hamiltonian
problem ($k$-LH), is defined as follows: Given a Hamiltonian $H=\sum_{i=1}^m H_{i}$
acting on $n$ qubits, such that each $H_{i}$ is Hermitian positive semidefinite of norm at most $1$ and acts non-trivially
on at most $k$ qubits, decide whether the smallest eigenvalue of
$H$ is smaller than some threshold $\alpha$ or larger than $\beta$ (where
$\beta-\alpha =\Omega(\poly^{-1}(n))$. The $5$-local Hamiltonian
problem was proven to be complete for the complexity class $\mathsf{QMA}$
by Kitaev in 1999~\cite{kitaev2002classical}; this result plays the same foundational role as the Cook-Levin theorem in classical complexity theory. Following this
result, it was shown that the $3$- and even $2$-local Hamiltonian problems are $\mathsf{QMA}$-complete \cite{kempe20033,kempe2006complexity}.

One can straightforwardly encode a classical $k$-CSP (acting on bits) as an instance of the $k$-local Hamiltonian problem by converting each classical constraint into a diagonal Hermitian operator $H_i$.   In this case, for all $i$ and $j$,
the terms $H_{i}$ and $H_{j}$ commute, i.e. $H_{i}H_{j}=H_{j}H_{i}$. The problem of 
approximating the ground energy of \emph{commuting} Hamiltonians (call it the $k$-CLH) thus lies between the $k$-CSP and $k$-LH. Intriguingly, commuting Hamiltonians seem to share features of both classical and quantum systems. It is known that ground states of commuting Hamiltonians can be highly entangled (e.g., Kitaev's toric code \cite{kitaev2003fault}). On the other hand, it was shown by Bravyi and Vyalyi that $2$-CLH for qudits is in $\mathsf{NP}$ \cite{bravyi2003commutative}. It was subsequently shown that the $3$-CLH for qubits and $4$-CLH for qubits on a grid are also in $\mathsf{NP} $\cite{aharonov2011complexity, schuch2011complexity}.  The complexity of $k$-CLH for general $k$ remains open.

In this paper we investigate the commutative version of a different problem related to local Hamiltonians. Specifically, we consider the ground state connectivity problem introduced in~\cite{gharibian2015ground}. 
Given a local Hamiltonian
$H$ on $n$ qubits and two quantum circuits that prepare 
starting and final states $\ket{\psi}$ and $\ket{\phi}$, is there a sequence of poly$(n)$ local
unitaries which maps $\ket{\psi}$ to $\ket{\phi}$ and is such that all intermediate
states have energy close to the ground energy of $H$? This problem is motivated by the analogous classical connectivity problems~\cite{gopalan2009connectivity}. Although it is not believed that $\mathsf{QCMA}$ is in $\mathsf{NP}$, the inclusion is consistent with our current state of knowledge. In particular, all known explicit constructions of commuting Hamiltonians have ground states that can be succinctly described by a classical witness. This motivates the introduction of the ground state connectivity problem as an attempt to define a ``hard'' problem for commuting Hamiltonians. An interesting class of local Hamiltonian for which the problem is relevant are quantum memories based on stabilizer codes~\cite{brown2016quantum}, which by definition have the property that they have efficiently preparable ground states that cannot be connected by a low-energy path, due to the presence of an energy barrier. 

It was shown
in \cite{gharibian2015ground} that the ground space connectivity problem is $\mathsf{QCMA}$-complete for $5$-local Hamiltonians. Here we show that the connectivity problem for commuting local Hamiltonians is just as hard.

\begin{theorem*}[Informal]
	Given an $n$-qubit commuting local Hamiltonian $H$,
	and quantum circuits which prepare two states $\ket{\psi}$ and $\ket{\phi}$
	in the ground space of $H$, it is $\mathsf{QCMA}$-complete to decide if there exists a $\poly(n)$-length sequence of $2$-local unitaries which maps $\ket{\psi}$ to $\ket{\phi}$ such that all intermediate states have low energy for $H$.
\end{theorem*}
	
To the best of our knowledge, this is the first problem in which commuting
local Hamiltonians exhibit complexity theoretic hardness equivalent to
that of general local Hamiltonians. 

To prove the theorem one must show both containment in, and hardness for, $\mathsf{QCMA}$. Containment in $\mathsf{QCMA}$ is straightforward, and follows from \cite{gharibian2015ground}. To prove $\mathsf{QCMA}$-hardness, we give a reduction from a simple problem which we call the $2$-layer problem. Roughly speaking this problem can be described as follows. One is given two local Hamiltonians $A$ and $B$ which are each composed of commuting local terms. One is asked to decide if $A+B$ has an efficiently preparable state with low energy, or all states have high energy for $A+B$ (promised that one of these possibilities holds).  We deduce $\mathsf{QCMA}$-completeness of this 2-layer problem using a simple reduction from $\mathsf{QCMA}$ circuit satisfiability. We include the proof in Appendix~\ref{appA}.\footnote{We note that it may also be possible to establish $\mathsf{QCMA}$-completeness of the 2-layer problem using known 1D circuit-to-Hamiltonian mappings such as the one from~\cite{aharonov2009power}.} Our main contribution is a Karp reduction from the 2-layer problem to the ground space connectivity problem, whereby a 2-layer Hamiltonian $A+B$  is mapped to a \textit{commuting} local Hamiltonian $H_{A}+H_{B}+G$. It is designed so that we can use
tools from \cite{gharibian2015ground} to analyze the resulting commuting Hamiltonian, after certain small modifications that enhance their generality. 

In section \ref{sec:Review-of-Linear}, we review relevant background, and define the computational problems of interest. In section \ref{sec:Hardness-of-Traversing},
we give the reduction which establishes $\mathsf{QCMA}$-completeness of the ground state connectivity problem for commuting Hamiltonians. We conclude in section \ref{sec:Conclusion-and-Open}.

\section{\label{sec:Review-of-Linear} Preliminaries}

The reader is referred to \cite{watrous2011theory} for relevant background in linear
algebra and quantum information, and \cite{arora2009computational} for background in  complexity
theory. We give less standard notation and definitions here. 

We consider finite-dimensional Hilbert spaces over the complex field $\mathbb{C}$. The 
space of $n$ qudits is $(\mathbb{C}^{d})^{\otimes n}$. Throughout the
paper, for simplicity, we consider only qubits, for which $d=2$, though our results
hold for any constant qudit dimension $d$. 

For any operator
$A$, the \emph{operator norm} will be denoted as $\|A\|=\max_{\|x\|=1}\|Ax\|$.
The trace norm is $\|A\|_{\mathrm{tr}}=\mathrm{Tr}(\sqrt{A^{\dagger}A})$.

We use the following terminology. For an integer $1\leq k \leq n$, we say that a $n$-qubit linear operator is $k$-local if it can be written as $I \otimes L$ where $L$ acts on at most $k$ qubits. A \emph{$k$-local Hamiltonian} 
is $H=\sum_{i}H_{i}$ where each $H_i$ is a $k$-local Hermitian operator. We will always assume the normalization $\|H_i\|\leq 1$ and that $H_i$ is positive semidefinite for each $i$. The Hamiltonian is said to be \emph{commuting} when the $H_i$ pairwise commute. To emphasize that a given Hamiltonian may not be commuting we sometimes call it a \emph{general}, or \emph{non-commuting}, $k$-local Hamiltonian. 
We say that \textit{a $k$-local Hamiltonian is $a$-layered} if the $H_i$ can be partitioned into $a$ sets such that the terms within the same set commute. Thus a commuting Hamiltonian is $1$-layered. The \emph{ground energy} of a Hamiltonian $H$ is its smallest eigenvalue. A \emph{ground state} of a Hamiltonian $H$ is any unit vector in its \emph{ground space}, the eigenspace associated with the smallest eigenvalue of
$H$. We use $\ket{0}^{\otimes n}$ and $|0^n\rangle$ interchangeably. 

\begin{defn}[Efficiently preparable states] A family $\{\ket{\psi_n}\}$ of states is \textit{efficiently preparable} if for each $n$, there is a sequence of $\poly(n)$ 2-local	unitaries that can prepare $|\psi_n\rangle$ from $\ket{0}^{\otimes n}$.
\end{defn}
A \textit{promise language} is a pair $L=(L_{yes}; L_{no})$ with $ L_{yes},L_{no}\subseteq \{0,1\}^*$ and $L_{yes}\cap L_{no}=\emptyset$, where $L_{yes}$ are strings in the language and $L_{no}$ are strings not in the language.
\begin{defn}[QCMA]
We say that a promise language $L$ is in the complexity class $\mathsf{QCMA}$ if there exist polynomials $p,q$ and a deterministic Turing machine that on input $x\in\{0,1\}^*$ outputs in time $q(|x|)$ (where $|x|$ is the length of $x$) the description of a quantum circuit $Q_x$ acting on $p(|x|)$ qubits such that the following holds:
 \begin{eqnarray*}
x\in L & \Rightarrow & \exists y\in\{0,1\}^{p(|x|)},
\| \Pi_1Q_x|y\rangle \| \geq\frac{2}{3}\, , \\
x\not\in L & \Rightarrow & \forall y\in\{0,1\}^{p(|x|)},
\| \Pi_1Q_x|y\rangle \| \leq\frac{1}{3}\, , \end{eqnarray*}
where $\Pi_1=|1\rangle\langle1|\otimes I$ is the orthogonal projection onto the state $|1\rangle$ on the first qubit.
\end{defn}

The constants $\frac{2}{3}$ and $\frac{1}{3}$ are arbitrary and
can be amplified to $(1-2^{-\text{poly}(n)},2^{-\text{poly}(n)})$ using standard techniques (we will use this fact later on).
 
We finally define the ground space connectivity problem for commuting Hamiltonians. Given a $k$-local commuting Hamiltonian $H$ and two efficiently preparable states $|\psi\rangle$ and $|\phi\rangle$ in the ground space of $H$, the problem is to decide whether there exists a sequence of 2-local unitaries that takes the initial state $|\psi\rangle$ close to the final state $|\phi\rangle$ in a way such that all the intermediate states have low energy. Formally, the problem is stated as follows (adapted from \cite{gharibian2015ground}): 

\begin{defn}[$k$-$\cgscon(c,s)$]\label{def:kcgscon}
Given as input a $k$-local $n$-qubit commuting Hamiltonian $H$, $n$-qubit $\text{poly}(n)$-size quantum circuits $U_{\psi}$ and $U_{\phi}$ that prepare the states $|\psi\rangle$ and $|\phi\rangle$ in the ground space of $H$, and real numbers $0\leq c<s$, decide whether
\begin{itemize}
\item{(YES)} There exists an $m=\poly(n)$ and a sequence of $2$-local
unitaries $U_{1}$ to $U_{m}$ such that 
$$\max\Big\{\|U_\phi\ket{0^n} -U_{m}...U_{1}U_\psi\ket{0^n}\|,\,\langle\psi_{i}|H|\psi_{i}\rangle,\,1\leq i\leq m\Big\}\leq c,$$
where $|\psi_{i}\rangle=U_{i}...U_{1}U_\psi\ket{0^n}$, or
\item{(NO)} For \emph{any} sequence of $m=\poly(n)$ $2$-local unitaries $U_{1}$
to $U_{m}$, 
$$\max\Big\{\|U_\phi\ket{0^n} -U_{m}...U_{1}U_\psi\ket{0^n}\|,\,\langle\psi_{i}|H|\psi_{i}\rangle,\,1\leq i\leq m\Big\}\geq s.$$
where $|\psi_{i}\rangle=U_{i}...U_{1}U_\psi\ket{0^n}$,
\end{itemize}
given the promise that one of the cases holds.
\end{defn}

We define the 2-layer problem as follows:

\begin{defn}[2-layer problem]\label{def:ABproblem} 
Given as input a pair of $n$-qubit  Hamiltonians $A$ and $B$, each consisting of a sum of mutually commuting  $15$-local terms (but terms appearing in the decomposition of $A$ do not necessarily commute with terms appearing in the decomposition of $B$),  $0\leq A\leq I$ and $0\leq B\leq I$ and parameters $0\leq \alpha < \beta $ such that $\alpha\leq 2^{-\poly(n)}$, $\beta \geq \poly^{-1}(n)$, decide whether
\begin{itemize}
\item (YES) There exists an efficiently preparable state $\ket{\psi}$ such that $\bra{\psi} A+ B \ket{\psi} \leq \alpha$, or
\item (NO) For all states $\ket{\phi}$, $\bra{\phi} A+ B \ket{\phi} \geq \beta$,
\end{itemize}
given the promise that one of the cases holds.
\end{defn}
In appendix \ref{appA} we prove that the 2-layer problem is $\mathsf{QCMA}$-complete using a reduction from $\mathsf{QCMA}$ circuit satisfiability.
\begin{lem}\label{lem:2layer}
The $2$-layer problem is $\mathsf{QCMA}$-complete.
\end{lem}
\section{\label{sec:Hardness-of-Traversing}QCMA-completeness of \cgscon}

In this section we consider the problem of traversing the ground
space of commuting Hamiltonians with constant locality. Informally, as stated
in the introduction, we want to understand the hardness of going from
one low energy state of a commuting Hamiltonian to another by applying a polynomial-length
sequence of two-qubit unitaries,
while maintaining low energy in each intermediate state. Our main result is that the ground space connectivity problem for commuting local
Hamiltonians is $\mathsf{QCMA}$-complete:

\begin{theorem}\label{maintheorem}
$21$-$\cgscon(2^{-\poly(n)},\poly^{-1}(n))$ is $\mathsf{QCMA}$-complete.
\end{theorem}

To prove this result we use many of the ideas from~\cite{gharibian2015ground}, where it is shown that the ground space connectivity problem for general (not necessarily commuting) local Hamiltonians is $\mathsf{QCMA}$-complete. The main part of the proof is a reduction from the 2-layer problem which establishes $\mathsf{QCMA}$-hardness.

We begin by discussing one of the main tools used in \cite{gharibian2015ground},
the Traversal Lemma. We modify and slightly generalize the lemma to suit our needs. We first review the definition of $k$-orthogonality and establish a ``small projection lemma'' which we use to prove the Modified Traversal Lemma.

\begin{defn}
[$k$-orthogonality \cite{gharibian2015ground}] For $k\geq1$, a pair of states $|\psi\rangle$,$|\phi\rangle$$\in(\mathbb{C}^{d})^{\otimes n}$
is $k$-orthogonal if for all $k$-local unitaries $U$, we have $\langle\psi|U|\phi\rangle=0$.
Further, any two subspaces $S,T\subseteq(\mathbb{C}^{d})^{\otimes n}$
are $k$-orthogonal if every pair of states $|\psi\rangle$,$|\phi\rangle$
in $S$ and $T$ respectively are $k$-orthogonal. For example, the
states $|000\rangle$ and $|111\rangle$ are 2-orthogonal.
\end{defn}

\begin{lem}
(Small Projection Lemma)\label{lem:(Small-Projection-Lemma)} Let
$S$ and $T$ be k-orthogonal subspaces and let $P_S$ and $P_T$ be the orthogonal projections onto them. Let $\Pi=I-P_S-P_T$.
 Let $|\psi_{0}\rangle \in S$.
Let $U_{1},\ldots,U_{m}$ be any sequence of k-local unitaries, and
let $|\psi_{i}\rangle=U_{i}|\psi_{i-1}\rangle$ for $1\leq i\leq m$.
Assume that for every $0\leq i\leq m$, $\|\Pi|\psi_{i}\rangle\|\leq\zeta$ for some $\zeta\geq0$. 
Then, for every $0\leq i\leq m$, $\|P_T|\psi_{i}\rangle\|\leq i\zeta$
and $\|P_S|\psi_{i}\rangle\|\geq 1-(i+1)\zeta$.\end{lem}

\begin{proof}
We show by induction that for every $0\leq i\leq m$, 
$\| P_T|\psi_{i}\rangle\|\leq i\zeta$.
It is trivially true for $i=0$. Further,
\begin{eqnarray*}
\|P_T|\psi_{i}\rangle\| 
 & = & \|P_TU_{i}(P_S+ P_T+\Pi)|\psi_{i-1}\rangle\|\\
 & \leq & \| P_TU_{i}\Pi|\psi_{i-1}\rangle\|+\|P_TU_{i} P_T|\psi_{i-1}\rangle\|\\
 & \leq & \|\Pi|\psi_{i-1}\rangle\|+\|P_T|\psi_{i-1}\rangle\|\\
 & < & \zeta+(i-1)\zeta.
\end{eqnarray*}
Here, 
the second line follows from the triangle inequality and the fact that
$P_TU_{i}P_S|\psi\rangle=0$
since the $U_{i}$ are $k$-local and the subspaces $S$ and $T$ are $k$-orthogonal, the third
line uses $\|P\|\leq1$, and the fourth the condition $\|\Pi|\psi_i\rangle\|<\zeta$ for $0\leq i\leq m$ and the induction hypothesis. 
Furthermore, by the triangle inequality, for all $0\leq i \leq m$, 
\[ \|P_S|\psi_{i}\rangle\| \geq \||\psi_{i}\rangle\|-\|P_T|\psi_{i}\rangle\|-\|\Pi|\psi_{i}\rangle\|\geq 1-(i+1)\zeta. \]
\end{proof}

We are now ready to prove a modified version of the Traversal Lemma
from \cite{gharibian2015ground} that we require in the next
section. The main modification compared to the original lemma is that ours considers an additional pair $(U,V)$ of $k$-orthogonal subspaces that represent an excluded (high-penalty) subspace for the traversal from one space to the other in the original pair $(S,T)$. 
% for pairs of 2-orthogonal subspaces with specific properties (stated in Lemma \ref{lem:(Modified-Traversal-Lemma)}) instead of only one pair of 2-orthogonal subspaces. 
Intuitively, the Modified Traversal Lemma states the following. Let $|\psi\rangle$ and $|\phi\rangle$ be any two $k$-orthogonal states 
which are in the ground space of some Hamiltonian $H$ and are contained in some subspace $U$. Let $V$ be a subspace such that $U$ and $V$ are $k$-orthogonal. Assume further that there is a high energy penalty for any state that is outside the direct sum $U + V$. Then, for any sequence of $m$ unitaries that map $|\psi\rangle$ to a state that is $\epsilon$-close to $|\phi\rangle$ in a manner such that every intermediate state has very little overlap with the orthogonal complement of $U + V$, there is some intermediate state that has sufficiently large overlap on the orthogonal complement of the union of the subspaces spanned by $|\psi\rangle$ and $|\phi\rangle$. 

\begin{lem}
(Modified Traversal Lemma) \label{lem:(Modified-Traversal-Lemma)}
Let $S$ and $T$ be two $k$-orthogonal subspaces, and let $U$ and $V$ be another pair of $k$-orthogonal subspaces such that $S\cup T\subseteq U$. Let $P_S,P_T,P_U,P_V$ be the orthogonal projections onto them, and let $\Pi=I-P_U-P_V$, $Q=P_S+P_T$, and $P=P_U-Q$. 
Let $|\psi_{0}\rangle\in S$
and $|\phi\rangle\in T$ be a pair of $k$-orthogonal states. 
Let $U_{1},\ldots,U_{m}$ be a sequence of $k$-local unitaries that map
$|\psi_{0}\rangle$ to $|\psi_{m}\rangle$, where 
$|\psi_{i}\rangle=U_{i}|\psi_{i-1}\rangle$. 

Let $\epsilon,\delta \geq0$ be such that $\||\phi\rangle-|\psi_{m}\rangle\|\leq\epsilon$ and $\|\Pi|\psi_{i}\rangle\|\leq \delta$ for $1\leq i\leq m$. Then there exists an $i\in \{1,\ldots, m\}$ such that $\|P|\psi_{i}\rangle\|^{2}\geq \big(\frac{1-\eps}{m}\big)^2-2(m+1)\delta$.
\end{lem}
 
\begin{proof}
Using the assumption $\|\Pi|\psi_{i}\rangle\|\leq \delta$ and the Small Projection Lemma (Lemma~\ref{lem:(Small-Projection-Lemma)}) with the subspaces $U$ and $V$, we get that for all $i\in\{1,\ldots,m\}$, $\|P_U|\psi_{i}\rangle\|\geq 1-(m+1)\delta$. Thus, letting $\zeta = \max_{1\leq i\leq m} \|P\ket{\psi_i}\|^2$, we have that for every $i\in \{1, \ldots, m\}$, 
\begin{align}
\langle\psi_{i}|Q|\psi_{i}\rangle
&=\langle\psi_{i}|P_U|\psi_{i}\rangle-\langle\psi_{i}|P|\psi_{i}\rangle\notag\\
&= \big\| P_U|\psi_{i}\rangle\big\|^2 - \big\|P|\psi_{i}\rangle\big\|^2\notag\\
& \geq (1-(m+1)\delta)^2-\zeta\notag\\
&\geq 1-2(m+1)\delta-\zeta.\label{eq:trq}
\end{align}
Let $\ket{\psi'_0} = \ket{\psi_0}$, and define $\ket{\psi'_i}$, for $1\leq i \leq m$, by induction as $\ket{\psi'_i} = QU_i \ket{\psi'_{i-1}}$. We prove by induction on $i$ from $0$ to $m$ that
\begin{equation}\label{eq:ind-hyp}
 \big\| \ket{\psi'_i} - \ket{\psi_i}\big\| \leq i\sqrt{2(m+1)\delta+\zeta}.
\end{equation}
The case $i=0$ is clear. Use the triangle inequality to write 
\begin{align*}
\big\| \ket{\psi'_i} - \ket{\psi_i}\big\| &=\big\| QU_i\ket{\psi'_{i-1}} - \ket{\psi_{i}}\big\| \\
&\leq\big\| QU_i(\ket{\psi'_{i-1}} - \ket{\psi_{i-1}})\big\| + \big\| (Q-I) \ket{\psi_i} \big\|\\
&\leq (i-1)\sqrt{2(m+1)\delta+\zeta} + \sqrt{2(m+1)\delta+\zeta},
\end{align*}
where the last inequality uses $\|QU_i\|\leq 1$ and the induction hypothesis for the first term, and~\eqref{eq:trq} for the second: using that $(I-Q)$ is a projection, 
$$\big\| (Q-I) \ket{\psi_i}\big\|^2 = \bra{\psi_i} (I-Q) \ket{\psi_i} \leq 2(m+1)\delta+\zeta.$$

Now observe that $|\psi'_{m}\rangle$ always lies in $S$. This is because
inductively, $|\psi_{0}\rangle$ lies in $S$, and if $|\psi'_{i-1}\rangle$
lies in $S$, then  \[|\psi'_i\rangle = QU_i|\psi'_{i-1}\rangle = (P_S+P_T)U_i|\psi'_{i-1}\rangle = P_SU_i|\psi'_{i-1}\rangle, \]
where for the last equality we used that for any $k$-local unitary $U_{i}$, $P_TU_i|\psi'_{i-1}\rangle=0$ since the subspaces $S$ and $T$ are $k$-orthogonal. But by assumption $\ket{\phi} \in T$ and thus $\langle \psi'_m | \phi \rangle$=0. Also $\|\ket{\psi_m}-\ket{\phi}\|\leq \eps$, and therefore by~\eqref{eq:ind-hyp} for $i=m$ we deduce that 
\[
1 \leq 
\| |\psi'_m\rangle - |\phi\rangle \| \leq
\| |\psi'_m\rangle - |\psi_m\rangle \| + \| |\psi_m\rangle - |\phi\rangle \| \leq
 m\sqrt{2(m+1)\delta+\zeta} + \eps,
\]
proving the lemma.
\end{proof}

The main step in our analysis is given by the following reduction from the $2$-layer problem to $\cgscon$.

\begin{theorem}[Reduction from 2-layer problem] \label{thm:2layertraversal}
Let $A$ and $B$ be $n$-qubit local Hamiltonians with $0\leq A\leq I$ and $0\leq B\leq I$ such that $A$ is a sum of commuting local terms and so is $B$. There exists an $(n+6)$-qubit Hamiltonian $H$, which is a sum of commuting local terms, each of which can be computed in time polynomial in $n$ from $A$ and $B$, and which has the following properties. The states $|\psi_0\rangle=|0^n\rangle|000\rangle|000\rangle$ and $|\phi\rangle=|0^n\rangle|111\rangle|000\rangle$ are zero energy ground states of $H$.  Moreover, for any $\alpha$, $\beta$ such that $0\leq \alpha < \beta \leq 2$, the following hold. 
\begin{itemize}
\item{ \textbf{(Completeness)} Suppose $\langle\psi|(A+B)|\psi\rangle\leq\alpha$ for some efficiently preparable state $|\psi\rangle$. Then there exists a sequence of 2-local unitaries $U_1,\ldots,U_m$, for $m=\poly(n)$, such that 
\[
U_mU_{m-1}\ldots U_1|\psi_0\rangle=|\phi\rangle
\]
and 
\[\langle \psi_i|H|\psi_i \rangle \leq \frac{1}{2}\alpha, \qquad \forall\,1\leq i\leq m,
\]
where $|\psi_{i}\rangle=U_{i}|\psi_{i-1}\rangle$.}
	
\item{\textbf{(Soundness)} Suppose $\langle\phi|(A+B)|\phi\rangle\geq\beta$ for all $|\phi\rangle$. Then for any sequence of 2-local unitaries $U_1,\ldots,U_m$ satisfying 
\[
\|U_mU_{m-1}\ldots U_1|\psi_0\rangle-|\phi\rangle\|\leq \frac{1}{2},
\]
we have
	\[\max\big\{\langle \psi_i|H|\psi_i \rangle\mbox,\, 1\leq i\leq m \big\} = 	\Omega\left(\frac{1}{m^6}\beta^2\right).\]
}
\end{itemize}
\end{theorem}

\begin{proof}
We first describe the construction of the commuting Hamiltonian $H$. Define $3$-qubit projectors
$$ P_{0}=|000\rangle\langle000|, \qquad P_{1}=|111\rangle\langle111|,\qquad \Pi=I-P_{0}-P_{1},$$
and
$$ P_{+}=\frac{1}{2}(|000\rangle+|111\rangle)(\langle000|+\langle111|),\qquad	P_{-}=\frac{1}{2}(|000\rangle-|111\rangle)(\langle000|-\langle111|).$$
Consider a system of $n+6$ qubits partitioned into three registers, where the first register contains $n$ qubits and the second and the third registers contain 3 qubits each.
Define 
\[
H=H_A+H_B+G,
\]
where
$$
H{}_{A}=A\otimes\Pi\otimes P_{+},\qquad
H{}_{B}=B\otimes\Pi\otimes P_{-},\qquad\text{and}\qquad
G=I\otimes I\otimes\Pi.
$$
Since $A$ is a sum of local commuting terms, $H_A$ is also a sum of local commuting terms. Likewise for $B$ and $H_B$. The projector $G$ is itself a $3$-local term. Since $P_{+}$ and $P_{-}$ are orthogonal, we see that the local terms in $H_A$ commute with the local terms in $H_B$. Moreover, all the local terms commute with $G$. Thus $H$ is a commuting local Hamiltonian. It is easily verified that  $|\psi_0\rangle=|0^n\rangle|000\rangle|000\rangle$ and $|\phi\rangle=|0^n\rangle|111\rangle|000\rangle$ are zero energy ground states of $H$. Further, it is easy to create $H$ from $A$ and $B$ since each of the projectors $P_{+}, P_{-}, \Pi$ acts only on 3 qubits. 

The intuition behind the construction is as follows. Given Hamiltonians $A$ and $B$ acting on $n$ qubits, a simple way to make them commute is to make their ranges orthogonal to each other. This can be achieved by letting $H_A = A\otimes W_1$ and $H_B = B \otimes W_2$, where $W_1$ and $W_2$ are any orthogonal projectors on $r$ ancilla qubits. However, this makes it simple for any traversal to pass only through the ground space of one of the Hamiltonians $A$ or $B$: by remaining in a state that is in the nullspace of $W_1$ or $W_2$ at all times. 

To prevent this, fix a state $|\Gamma\rangle$ on the ancilla qubits such that $\langle\Gamma|W_1|\Gamma\rangle \geq c$ and $\langle\Gamma|W_2|\Gamma\rangle \geq c$ for some constant $c>0$, and introduce a penalty Hamiltonian $G=I\otimes(I-|\Gamma\rangle\langle\Gamma|)$. Concretely, we set $r=3, W_1 = P_+, W_2 = P_-, |\Gamma\rangle=|000\rangle$, which gives $c=1/2$ and $G=I\otimes(I-|000\rangle\langle 000|)$. 

This fixes our ``trivial traversal'' issue, but $H_A$, $H_B$ and $G$ no longer commute. The key insight then is to notice that we can use $G=I\otimes(I-|000\rangle\langle 000|-|111\rangle\langle 111|)$ instead. This choice of $G$ commutes with both $H_A$ and $H_B$, and penalizes any state on the ancilla qubits that is not in the direct sum of the subspaces spanned by $|000\rangle$ and $|111\rangle$. 

A final issue now is that a unitary could map the state $|000\rangle$ on the ancilla qubits to e.g. $|GHZ\rangle = \frac{1}{\sqrt{2}}(|000\rangle+|111\rangle)$, which is in the ground space of $G$, but trivially also in the ground space of $P_-$, and thus the traversal can ignore $B$ completely, a problem similar to the one we started with. Here we use $2$-orthogonality: since the states $|000\rangle$ and $|111\rangle$ are $2$-orthogonal, no $2$-local unitary can map $|000\rangle$ to $|111\rangle$, and thus any sequence of states that maps $|000\rangle$ to $|GHZ\rangle$ on the ancilla qubits must pass through a state on those qubits that $G$ penalizes, ensuring that the construction is sound. 

Below we formally establish the required completeness and soundness properties.

\vspace{-0.5em}
\paragraph{Completeness.}
Suppose $\langle\psi|(A+B)|\psi\rangle\leq\alpha$ for some efficiently preparable $|\psi\rangle$. Let $C$ be a circuit consisting of $m'= \mbox{poly}(n)$ $2$-local unitaries which prepares $|\psi\rangle$ starting from $|0^n\rangle$. Consider the following sequence of $2$-local unitaries which maps $\ket{\psi_0}$ to $\ket{\phi}$:
\begin{enumerate}
\item{ Starting from $|\psi_0\rangle=|0^n\rangle|000\rangle|000\rangle$, apply the circuit $C$ to the first register.  (After this step the state is $|\psi\rangle|000\rangle|000\rangle$.)}
\item{Apply a sequence of single-qubit $X$ gates to flip the bits of the second register from $000$ to $111$. (After this step the state is $|\psi\rangle|111\rangle|000\rangle$.)}
\item{Apply the circuit $C^{\dagger}$ to the first register. (After this step the state is $|\phi\rangle=|000\rangle|111\rangle|000\rangle$.)}
\end{enumerate}

Note that at all times during steps 1. and 3. we remain in the zero energy ground space of $H$. At all times during step 2. the state is of the form $|\psi\rangle|a\rangle|000\rangle$ where $|a\rangle$ is a computational basis state of $3$ qubits. Its energy is bounded as
\begin{align*}
\langle\psi|\langle a|\langle000|(H_{A}+H_{B}+G)|\psi\rangle|a\rangle|000\rangle  
 & =  \big(\langle\psi|A|\psi\rangle \langle000|P_{+}|000\rangle +\langle\psi|B|\psi\rangle\langle000|P_{-}|000\rangle\big) \langle a|\Pi|a\rangle\\
 & \leq  \frac{1}{2}\langle\psi|A|\psi\rangle+\frac{1}{2}\langle\psi|B|\psi\rangle\\
 & \leq  \frac{1}{2}\alpha.
\end{align*}
The total number of $2$-local unitaries in the sequence is $m=2m'+3$.

\vspace{-0.5em}
\paragraph{Soundness.}
Suppose that $\langle\phi|(A+B)|\phi\rangle\geq\beta$ for all $|\phi\rangle$. Let
$U_1,\ldots,U_m$ be a sequence of $2$-local unitaries and let $|\psi_i\rangle=U_i|\psi_{i-1}\rangle$ for $1\leq i\leq m$. Suppose $\||\psi_{m}\rangle-|\phi\rangle\| = \epsilon$ for some $\epsilon\leq \frac{1}{2}$. Let $\delta = \max_{1\leq i \leq m} \|I\otimes I\otimes\Pi|\psi_{i}\rangle\|$, and note that the energy
violation is at least $\max_{1\leq i \leq m} \langle\psi_{i}|G|\psi_{i}\rangle\geq \delta^{2}$. Applying the Small Projection Lemma (Lemma~\ref{lem:(Small-Projection-Lemma)}) with the $2$-orthogonal subspaces $S$ and $T$ defined as the +1 eigenspaces of the projectors
$P_S=I\otimes I\otimes P_0$ and
$P_T=I\otimes I\otimes P_1$ respectively, we obtain that
\begin{equation}\label{eq:soundness-1}
\|(I\otimes I\otimes P_{1})\ket{\psi_{i}}\| < i\delta
\end{equation}
for all $i$. 
Next apply the Modified Traversal Lemma (Lemma~\ref{lem:(Modified-Traversal-Lemma)}) with subspaces $S, T, U, V$  defined as the +1 eigenspaces of the  projectors $P_S=I\otimes P_{0}\otimes P_{0}$, $P_T=I\otimes P_{1}\otimes P_{0}$, $P_U=I\otimes I\otimes P_{0}$, and 
$P_V=I\otimes I\otimes P_{1}$ respectively. 
We obtain that there is some $i$ for which 
\begin{equation}\label{eq:lb-1}
\langle\psi_{i}|(I\otimes\Pi\otimes P_{0})|\psi_{i}\rangle\geq \Big(\frac{1-\eps}{m}\Big)^2-2(m+1)\delta.
\end{equation}
Then 
\begin{eqnarray*}
\langle\psi_{i}|H_{A}+H_{B}|\psi_{i}\rangle & = & \langle\psi_{i}|(A\otimes\Pi\otimes P_{+}+B\otimes\Pi\otimes P_{-})|\psi_{i}\rangle\\
 & = & \frac{1}{2}\langle\psi_{i}|(A+B)\otimes\Pi\otimes P_{0}|\psi_{i}\rangle+E\\
 & \geq & \frac{\beta}{2}\Big( \Big(\frac{1-\eps}{m}\Big)^2-2(m+1)\delta\Big) +E,
\end{eqnarray*}
where the last line uses the operator inequality $(A+B)\geq \beta I$ and~\eqref{eq:lb-1}, and $E$ is the error term defined below, which we bound next. Letting $P_{01}=|000\rangle\langle111|$ and $P_{10}=|111\rangle\langle000|$ we have
\begin{align*}
|E| & = \frac{1}{2}\big|\langle\psi_{i}|(A-B)\otimes\Pi\otimes P_{01}|\psi_{i}\rangle+\langle\psi_{i}|(A-B)\otimes\Pi\otimes P_{10}|\psi_{i}\rangle+\langle\psi_{i}|(A+B)\otimes\Pi\otimes P_{1}|\psi_{i}\rangle\big|\\
 & \leq \frac{1}{2} \big(\|(I\otimes I\otimes P_{1})\psi_{i}\|\|A-B\|
 +\|(I\otimes I\otimes P_{1})\psi_{i}\|\|A-B\|
+\|(I\otimes I\otimes P_{1})\psi_{i}\|\|A+B\|\big)\\
 & \leq  2\|(I\otimes I\otimes P_{1})\psi_{i}\|\\
 & \leq  2m\delta,
\end{align*}
where the first inequality uses the triangle inequality, the Cauchy-Schwarz inequality, and the fact that the norm of a projection is less than $1$, the second inequality
uses $\|A\|\leq1$ and $\|B\|\leq1$, and the third uses~\eqref{eq:soundness-1}. Overall the energy lower bound for soundness is
$$ \max_{1\leq i\leq m} \langle \psi_i|H|\psi_i\rangle \geq \max\Big\{ \delta^2, \,\frac{\beta}{2}\Big( \Big(\frac{1-\eps}{m}\Big)^2-2(m+1)\delta\Big)-2m\delta\Big\}.$$ 
To conclude it suffices to verify  that for any $0\leq \delta \leq 1$, this is $\Omega(\beta^2/m^6)$, as claimed. If $\delta = \Omega(\beta/m^3)$, then recall that since $\beta \leq 2$ due to our normalization of $A$ and $B$, the statement is true. If $\delta = O(\beta/m^3)$ for a small enough implicit constant, the second expression in the max is dominated by $(\beta/2)((1-\eps)/m)^2 = \Omega(\beta^2/m^2) = \Omega(\beta^2/m^6)$. 

\end{proof}

With Theorem~\ref{thm:2layertraversal} in hand, $\mathsf{QCMA}$-hardness of
$\cgscon$ follows from $\mathsf{QCMA}$-hardness of the $2$-layer problem.

\begin{proof}[Proof of Theorem \ref{maintheorem}]
Since commuting local Hamiltonians are a subset of general local Hamiltonians, the containment $21$-$\cgscon(2^{-\poly(n)},\poly^{-1}(n)) \in \mathsf{QCMA}$ follows from \cite{gharibian2015ground}. 
To see that the problem is $\mathsf{QCMA}$-hard, let $A+B$ be an instance of the $2$-layer problem with completeness and soundness parameters $c,s$ constructed from the $\mathsf{QCMA}$ verifier circuit for some language $L\in\mathsf{QCMA}$ according to Lemma \ref{lem:2layer}. 
Construct the Hamiltonian  $H=H_{A}+H_{B}+G$ from $A+B$ as in the proof of Theorem \ref{thm:2layertraversal}. Then $H$ is a $21$-local commuting Hamiltonian. Moreover, the theorem states that the 2-layer problem instance $A+B$ has the same solution (yes/no) as the $\cgscon$ instance specified by $H$, its ground states $|\psi_0\rangle,|\phi\rangle$, completeness parameter $c/2=2^{-\text{poly}(n)}$ and soundness parameter $\Omega(s^2/m^6)=\poly^{-1}(n)$.  Since the $2$-layer problem is $\mathsf{QCMA}$-hard (by lemma \ref{lem:2layer}), this establishes that $21$-$\cgscon(2^{-\poly(n)},\poly^{-1}(n))$ is also $\mathsf{QCMA}$-hard.
\end{proof}

\section{\label{sec:Conclusion-and-Open}Conclusions}

We have shown that the ground space connectivity problem for commuting local
Hamiltonians is $\mathsf{QCMA}$-complete, strengthening the previous result \cite{gharibian2015ground} for general local Hamiltonians. Thus, the commutativity restriction does not lower the complexity of this problem, as one might have a priori expected. It remains open to determine whether there exist instances of the commuting local Hamiltonian problem that are $\mathsf{QMA}$-hard. 

\section*{Acknowledgments}
T.V. is~supported by NSF CAREER grant CCF-1553477 and an AFOSR YIP award number FA9550-16-1-0495. All authors were partially supported by the IQIM, an NSF Physics Frontiers Center (NSF Grant PHY-1125565) with support of the Gordon and Betty Moore Foundation (GBMF-12500028).  

\bibliographystyle{alpha}
\bibliography{commcon}

\appendix

\section{Proof of Lemma \ref{lem:2layer}}\label{appA}

In this appendix we prove that the 2-layer problem is $\mathsf{QCMA}$-complete. 

Below we use the term ``$(k,l)$-local Hamiltonian'' to describe a $k$-local Hamiltonian such that each qubit is acted on nontrivially by at most $l$ terms.

The first step consists in applying a slight variant of Kitaev's circuit to Hamiltonian construction, which ensures that the Hamiltonian constructed from any $\mathsf{QMA}$ or $\mathsf{QCMA}$ verification circuit is not only local but also such that every qubit is acted (nontrivially) upon by a constant number of local Hamiltonian terms. In particular, it is a $(5,4)$-local Hamiltonian. This step is crucial, as without it our construction would eventually result in an $\Omega(\log n)$-local commuting Hamiltonian instead of a commuting Hamiltonian with constant locality. 

The idea to achieve this is simple: we introduce SWAP gates to ensure that every qubit is acted upon by at most $3$ unitaries: swapping in the state of a given qubit,
applying a circuit unitary, and swapping the state out to some new location. Once the circuit has been modified in this way, Kitaev's construction, which converts a circuit consisting of $2$-local unitaries to a $5$-local Hamiltonian through the use of a unary clock, yields a Hamiltonian with the desired property. We omit the details,
and the interested reader is referred to Section 2 in  \cite{oliveira2008complexity}.
We state the result in the following lemma.

\begin{lem}
	\label{lem:5,4-lhp}Let $C=U_{m}...U_{1}$ be an $n$-qubit quantum circuit where $m=$poly$(n)$ such that each qubit is acted upon by at most three gates. Then there exists a 
	 $(5,4)$-local Hamiltonian $H=\sum_{i=1}^{m'} H_{i}$ acting on $n'=\poly(n,m)$
	qubits such that the following holds:
	\begin{enumerate}
		\item If there exists a state $|\psi\rangle$ such that $C$ accepts $|\psi\rangle$ with
		probability $1-\epsilon$, then $\langle\phi|H|\phi\rangle\leq\frac{\epsilon}{\poly(m)}$,
		where 
		\begin{equation}
		|\phi\rangle=\frac{1}{\sqrt{m+1}}\sum_{i=0}^{m}U_{i}...U_{1}|\psi\rangle\otimes|0\rangle\otimes|i\rangle,
		\label{eq:hist}
		\end{equation}
		with $i$ written in unary notation.
		\item If $C$ accepts any state $|\psi\rangle$ with probability at most $\epsilon$,
		then for all $|\psi\rangle$, $\langle\psi|H|\psi\rangle\geq\frac{1-\epsilon}{\poly(m)}$.
	\end{enumerate}
\end{lem}

By starting with an error-amplified $\mathsf{QCMA}$ verifier for which the $\epsilon$ parameter is exponentially small, we can make the completeness and soundness parameters  $2^{-\text{poly}(n)}$ and  $\frac{1}{\text{poly}(n)}$ respectively in Lemma \ref{lem:5,4-lhp}. The second step of our construction consists in showing how the resulting instance of the $(5,4)$-local Hamiltonian
can be converted into a 2-layered local Hamiltonian such that all terms inside each of the layers commute, and are $15$-local. The completeness and soundness parameters will remain of order $2^{-\text{poly}(n)}$ and $\frac{1}{\text{poly}(n)}$ respectively.

We explain the intuition behind the construction. The first step is to individually modify each of the local terms from the input Hamiltonian instance into local terms such that the new Hamiltonians have orthogonal range. This results in a local Hamiltonian $G$ that is trivially commuting. However, it is now easy for a prover to specify a state in the simultaneous ground space of all the local terms, irrespective of whether the initial instance was satisfiable or not. 

To prevent this we create another Hamiltonian $R$ that forces any
ground state to have high overlap with the ground space of
each Hamiltonian in $G$. Specifically, consider a qubit $j\in\{1,\ldots,n\}$, and let $H_{i_1(j)},H_{i_2(j)},H_{i_3(j)},H_{i_4(j)}$ be the local terms acting on it. We first allocate an ancilla qudit of dimension $4$ (equivalently, two qubits) $j'\in\{1,\ldots,n\}$ for $j$. We then create
four new Hamiltonians $H_{i_t(j)}'=H_{i_t(j)}\otimes|t-1\rangle\langle t-1|$ for $t\in\{1,\ldots,4\}$,
where the second projector acts on $j'$, and $G=H'_{i_1(j)}+H'_{i_2(j)}+H'_{i_3(j)}+H'_{i_4(j)}$. Note that a malicious prover can now easily create a state in the simultaneous null eigenspace of $H'_{i_2(j)},H'_{i_3(j)},H'_{i_4(j)}$
by creating a state that is $|00\rangle$ on $j'$. To prevent this, we introduce a penalty Hamiltonian $R=I\otimes(I-|\gamma\rangle\langle\gamma|)$,
where $|\gamma\rangle=\frac{1}{2}\sum_{i=0}^{3}|i\rangle$. Provided a large enough energy penalty is associated with $R$, any state in a low eigenspace of the $H_{i_t(j)}'$ is forced
to be close to the null eigenspace of $R$ on $j'$, namely,
by being in the state $|\gamma\rangle$, effectively translating the commuting terms in $G$ into a sum of the non-commuting original terms. 
In this case, there is a constant overlap with the projectors
of each of $H'_{i_t(j)}$, $t\in\{1,\ldots,4\}$ on the Hilbert space of
the qudit $j'$, and thus none of the Hamiltonians
can be trivially annihilated.

\begin{lem}
(Ground Space Preserving Layers) \label{lem:(Ground-Space-Preserving}
Let $H=\sum_{i=1}^{m}H_{i}$ be a $(5,4)$-local Hamiltonian acting on a system of $n$ qubits numbered
from $1$ to $n$,
where $\|H_{i}\|\leq1$.
For each qubit $j\in \{1,\ldots,n\}$, introduce a $4$-dimensional ancilla qudit $j'$. Let $c$
and $s$ be two reals such that $s-c\geq m^{-b}$ (for some
$b>0$).
	
For $i\in\{1,\ldots,m\}$,  let $H_{i}$ act on qubits $\{j_1(i),\ldots,j_5(i)\}$. For $t\in\{1,\ldots,5\}$, let $p_t(i)\in\{1,\ldots,4\}$ be such that $H_i$ is the $p_t(i)$-th Hamiltonian acting on $j_t(i)$. 
Define $G_{i}$
acting on qubits $\{j_{1}(i),...,j_{5}(i)\}$ and $\{j'_1(1),\ldots,j'_5(i)\}$ as
\[
G_{i}=H_{i}\otimes|p_{1}(i)-1\rangle\langle p_{1}(i)-1|\otimes...\otimes|p_{5}(i)-1\rangle\langle p_{5}(i)-1|.
\]
Let $|\gamma\rangle=\frac{1}{2}\sum_{i=0}^{3}|i\rangle$, and for every qubit
$j\in\{1,\ldots,n\}$ define $R_{j}$ acting on $j'$ 	as 
\[
R_{j}=I-|\gamma\rangle\langle\gamma|.
\]
Let $G=\sum_{i=1}^{m}G_{i}$ and $R=m^{r}\sum_{j=1}^{n}R_{j}$, where $r=b+5$. Then
the following hold, for $\kappa=2^{10}$: 
	\begin{enumerate}
		\item If there exists a state $|\psi\rangle$ such that $\langle\psi|H|\psi\rangle\leq c$,
		then the state $|\psi'\rangle=|\psi\rangle\otimes |\gamma\rangle^{\otimes n}$ satisfies $\langle\psi'|(G+R)|\psi'\rangle\leq\frac{c}{\kappa}$.
		\item If for all states $|\psi\rangle$, $\langle\psi|H|\psi\rangle\geq s$,
		then for all states $|\psi'\rangle$ (on the extended system), $\langle\psi'|(G+R)|\psi'\rangle\geq\frac{s}{\kappa}-\frac{1}{m^{b+1}}$.
	\end{enumerate}

\end{lem}

\begin{proof}

\emph{(Completeness) }Let $|\psi\rangle$ be such that $\langle\psi|H|\psi\rangle\leq c$.
Define $|\psi'\rangle=|\psi\rangle\otimes|\phi\rangle$, where $|\phi\rangle=|\gamma\rangle\otimes...\otimes|\gamma\rangle$
($n$ times) and where $|\psi\rangle$ is a state on the initial $n$ qubits
and each copy of $|\gamma\rangle$ acts on one of the ancilla qudits. By direct calculation, 

\begin{eqnarray*}
	\langle\psi'|(G+R)|\psi'\rangle & = & \langle\psi|\langle\phi|G|\psi\rangle|\phi\rangle+\langle\psi|\langle\phi|R|\psi\rangle|\phi\rangle\\
	& = & \sum_{i=1}^{m}\langle\psi|H_{i}|\psi\rangle|\langle\gamma|p_{1}(i)-1\rangle|^{2}...|\langle\gamma|p_{5}(i)-1\rangle|^{2}+0\\
	& = & \sum_{i=1}^{m}\frac{1}{4^{5}}\langle\psi|H_{i}|\psi\rangle\\
	& \leq & \frac{c}{\kappa}.
\end{eqnarray*}
	
\emph{(Soundness)} 
Let $|\phi_0\rangle, |\phi_1\rangle, \ldots, |\phi_d\rangle $, for $d=2^{2n}-1$, be an orthonormal basis on the $n$ ancilla qudits consisting of eigenvectors of $R$. Let $\lambda_0\leq \cdots \leq \lambda_{d}$ be the eigenvalues of $R$, counted with multiplicity. Then 
$|\phi_0\rangle=|\gamma\rangle^{\otimes n}$ and $\lambda_0=0$; moreover for $1\leq i \leq d$, $\lambda_i\geq m^r$. 
	
We proceed by contradiction. Let $|\psi'\rangle$ be a state such
that $\langle\psi'|(G+R)|\psi'\rangle\leq\frac{s}{\kappa}-\frac{1}{m^{b+1}}$.
Expand $|\psi'\rangle = \sum_{i=0}^d |a_i\rangle |\phi_i\rangle$, where $\{|a_0\rangle, \ldots, |a_d\rangle\}$ are vectors on the initial $n$ qubits. Let $|u\rangle = |a_0\rangle |\phi_0\rangle $ and $|v\rangle = \sum_{i=1}^d |a_i\rangle |\phi_i\rangle$, so that $|\psi'\rangle = |u\rangle + |v\rangle $. Let 
$$c_0 = \|\ket{u}\| = \|\ket{a_0}\|\qquad\text{and}\qquad c_1 = \|\ket{v}\| = \Big(\sum_{i=1}^d \|\ket{a_i}\|^2\Big)^{1/2}.$$
Since $\langle u, v\rangle = 0$, we have $1 = \||\psi'\rangle\|^2 = c_0^2 + c_1^2$. 
Then
\[
\langle\psi'|R|\psi'\rangle = \sum_{i=0}^d \lambda_i \| |a_i\rangle | \|^2  \geq
m^{r} c_{1}^{2}.
\]
Further, 
\begin{equation}\label{eq:g-1}
	\langle u|G|u\rangle =  
	\langle a_0|\langle \phi_0|G |a_0\rangle |\phi_0\rangle =
	\frac{1}{\kappa} \langle a_0|H|a_0\rangle \geq
	\frac{1}{\kappa} s\langle a_0|a_0\rangle =
	\frac{sc_0^2}{\kappa},
\end{equation}
where the inequality uses the soundness condition $\forall|\psi\rangle,\langle\psi|H|\psi\rangle\geq s \langle\psi|\psi\rangle$. Also note that 	
\[
	|\langle u|G|v\rangle| =
	|\langle v|G|u\rangle| \leq 
	\|v\|\|Gu\| \leq m\|v\|\|u\| \leq mc_1, 
\]
where the first inequality is Cauchy-Schwarz and the second uses the fact that $\|G\| \leq m$. Similarly, $|\langle v|G|v\rangle| \leq mc_1$, thus using~\eqref{eq:g-1},
\[
	\langle\psi'|G|\psi'\rangle \geq
	\langle u|G|u\rangle - |\langle u|G|v\rangle| - |\langle v|G|u\rangle| - |\langle v|G|v\rangle| \geq 
	\frac{sc_{0}^{2}}{\kappa}-3c_{1}m.
\]	
We then get 
\begin{eqnarray*}
	\langle\psi'|(G+R)|\psi'\rangle 
	& \geq & \frac{sc_{0}^{2}}{\kappa}-3c_{1}m+c_{1}^{2}m^{r}\\
	& = & \frac{s}{\kappa}(1-c_{1}^{2})-3c_{1}m+c_{1}^{2}m^{r}\\
	& \geq & \frac{s}{\kappa}-\frac{1.5^{2}m^{2}}{m^{r}-\frac{s}{\kappa}}\\
	& \geq & \frac{s}{\kappa}-\frac{1}{m^{b+2}},
\end{eqnarray*}
where the expression in the third line was obtained by setting $c_{1}=\frac{1.5m}{m^{r}-\frac{s}{\kappa}}$,
which minimizes the expression in the second line. 		
We have reached a contradiction with the assumption
$\langle\psi'|(G+R)|\psi'\rangle\leq\frac{s}{\kappa}-\frac{1}{m^{b+1}}$.
	
\end{proof}

\begin{rem*}
Any
Hamiltonian in layer $G$ acts on the $5$ qubits the original Hamiltonian
acted on, and further, the new Hamiltonian acts on one qudit or two ancilla qubits
for every qubit the original Hamiltonian acted on, increasing the locality by $10$, making
it $15$-local. Each of the Hamiltonians in $R$ act on two qubits. Further, all the Hamiltonians in the layers $G$ and $R$ commute within the layers.
\end{rem*}

Lemma~\ref{lem:2layer} is an  immediate corollary of Lemma~\ref{lem:5,4-lhp} and Lemma~\ref{lem:(Ground-Space-Preserving}.

\begin{proof}[Proof of Lemma~\ref{lem:2layer}]
The containment in $\mathsf{QCMA}$ is immediate since the prover can provide the classical description of the poly($n$)-sized quantum circuit consisting of 2-qubit unitaries that prepares the state $|\psi\rangle$ from $|0\rangle^{\otimes n}$ such that $\langle \psi|(A+B)|\psi \rangle\leq c$ if such an efficiently preparable state exists. 
To see $\mathsf{QCMA}$-hardness, let $V$ be the $\mathsf{QCMA}$ verifier for a language $L\in\mathsf{QCMA}$, $H=\sum_{i=1}^{m}H_{i}$ the $n$-qubit $(5,4)$-local Hamiltonian constructed from $V$ and an input $x$ to $L$ as in Lemma~\ref{lem:5,4-lhp}. Using standard error amplification for $\mathsf{QCMA}$, we may assume that the parameter $\eps$ in the lemma is $\eps=2^{-\poly(n)}$. 

Let $G_i$, $i\in\{1,\ldots,m\}$ and $R_j$, $j\in\{1,\ldots,n\}$, be the local projectors defined in Lemma~\ref{lem:(Ground-Space-Preserving}. Note that the $G_i$ (resp. $R_j$) are mutually commuting. Define 
\begin{equation}
\label{eq:AB}
A = \frac{1}{nm^{r}}\sum_i G_i \qquad \qquad B = \frac{1}{n}\sum_j R_j,
\end{equation}
where $r=O(1)$ is selected as in the lemma.  The normalizing factors in Equation \ref{eq:AB} are chosen to ensure that $0\leq A,B\leq I$, and so that $A+B$ is proportional to $G+R$ as defined in the lemma. The proportionality factor is $n^{-1}m^{-r}$ which is inverse polynomially small. $H'=A+B$ is a 2-layered $15$-local Hamiltonian whose ground state energy is exponentially small if $x\in L$ (i.e. $V$ accepts $x$) and at least inverse polynomially large if $x\notin L$ (i.e. $V$ rejects $x$).  Moreover, if $x\in L$ there is an efficiently preparable state which achieves exponentially small energy for $H$ (it is a tensor product of the history state~\eqref{eq:hist} and a symmetric product state).
\end{proof}

\end{document}